\documentclass[12pt,a4paper,reqno]{amsart}

\usepackage{amsmath,amssymb,amsthm,amsfonts}
\usepackage{amsbsy}
\usepackage{mathtools}

\usepackage[utf8]{inputenc}
\usepackage[english]{babel}
\usepackage{url}
\usepackage{wrapfig}
\usepackage{enumitem}
\usepackage{multicol}
\usepackage{moreenum}
\usepackage{subcaption}
\usepackage{xypic}
\usepackage{tikz}

\usepackage[pagewise]{lineno}

\usepackage[a4paper,left=25mm,right=25mm,top=30mm,bottom=30mm]{geometry}

\usepackage{xcolor}

\definecolor{LinkColor}{rgb}{0,0,100} 
\usepackage[colorlinks=true,linkcolor=LinkColor,citecolor=LinkColor,urlcolor= LinkColor, naturalnames, hyperindex, pdfstartview=FitH, bookmarksnumbered, plainpages]{hyperref}

\newtheorem{theorem}{Theorem}[section]

\newtheorem{lemma}[theorem]{Lemma}
\newtheorem{proposition}[theorem]{Proposition}

\theoremstyle{definition}

\newtheorem*{thx}{Acknowledgements}

\newcommand{\SL}{\operatorname{SL}}

\newcommand{\F}{\mathbb{F}}

\newcommand{\Z}{\mathbb{Z}}

\newcommand{\ZG}{\textup{Z}}

\newcommand{\Heis}{\textup{Heis}}

\newcommand{\Cyc}{\textup{C}}

\newcommand{\Hom}{\textup{Hom}}
\newcommand{\Maps}{\textup{Maps}}
\newcommand{\Sn}{\textup{S}}
\newcommand{\cl}{\textup{cl}}

\newcommand{\graffe}[1]{\left\lbrace #1\right\rbrace }

\newcommand{\gen}[1]{\langle #1\rangle}
\newcommand{\cor}[1]{\mathcal{#1}}

\newenvironment{example}[1][]{\refstepcounter{theorem}\par\medskip\noindent \textbf{Example~\thedefinition. #1} \rmfamily}{\medskip}

\title{Cryptographic multilinear maps using pro-$p$ groups}

\author{Delaram Kahrobaei and Mima Stanojkovski}

\date{\today}

\begin{document}

\begin{abstract}
In \cite{KTT/20}, the authors show how, to any nilpotent group of class $n$, one can associate a non-interactive key exchange protocol between $n+1$ users. The \emph{multilinear} commutator maps associated to nilpotent groups play a key role in this protocol. In the present paper, we explore some alternative platforms, such as  pro-$p$ groups.
\end{abstract}

\maketitle 







\section{Introduction}
In recent years and since the seminal paper by Boneh and Silverberg \cite{BonehSilverberg/03}, multilinear maps have become very popular in cryptography. Their use enables computations on hidden data and generalizes the ideas behind Diffie-Hellman's key exchange protocol.
Additionally, their applications are numerous: multiparty Diffie-Hellman key exchange, functional encryption, indistinguishability obfuscation. An excellent survey on the subject is \cite{survey/16}.

The concept of cryptographic multilinear maps generalizes that of bilinear pairings associated to elliptic curves: such multilinear maps are efficiently computable and may provide secure cryptosystems. In particular, first examples of multilinear maps in cryptography come from geometry, are defined on cyclic groups, and are bilinear. Constructing efficiently computable and cryptographically interesting maps is a hard task and to the best of our knowledge it is still unclear whether this can be done for every rank. 

There are three main constructions of cryptographic systems based on bilinear maps: the original one from Garg, Gentry and Halevi \cite{GGH/13}, a variant due to Coron, Lepoint and Tibouchi \cite{CLT/13}, and a “graph-induced” construction by Gentry, Gorbunov and Halevi \cite{GGH/15}.
Moreover, in a series of papers, Huang has presented geometric constructions of trilinear maps with a view towards cryptographic applications \cite{Huang/18, Huang/18II, Huang/19, Huang/20}. The main difficulty, as already many authors pointed out, is not to construct multilinear maps, but rather to check whether they can be efficiently employed for the production of secure cryptosystems. 

In \cite{KTT/20},  the authors construct a \emph{Non-Interactive Key Exchange} protocol (NIKE) employing multilinear maps within the realm of group theory. In particular, they exploit commutator maps  in nilpotent groups to construct key exchange protocols for a number of users equal to one more than the \emph{nilpotency class} of the platform group.
Note that nilpotent groups (in particular $p$-groups or more generally Engel groups) have been proposed for other cryptographic applications; we refer to \cite{KahNoc/20} for a survey and open problems.

In the present paper, we discuss the connection between the security of the cryptosystem from \cite{KTT/20} and the \emph{discrete logarithm problem} (DLP) in finite $p$-groups. We recall that, if $C=\gen{x}$ is a finite cyclic group and $y$ is an element of $C$, then solving the discrete logarithm problem for $y$ means to find an integer $a$ (equivalently the class of $a\bmod |C|$) such that $y=x^a$. In particular, the larger the order of the $p$-group $C$ is, the more work is likely to be required to solve the discrete logarithm problem in $C$. 

Motivated by the wish of being able to share a secret key between any number of users,
we propose, moreover, the employment of infinite pro-$p$ groups in the key exchange system in the following way. If $G$ is an infinite nonnilpotent pro-$p$ group and $n\geq 2$ is an integer, then $G$ has a finite quotient $\overline{G}$ of nilpotency class $n$ and so, over $\overline{G}$, we have a non-interactive key exchange protocol between $n+1$ users. In this sense, the group $G$ is a platform for an arbitrary number of users. We also show that groups like $G$ exist, with the additional property that they are ``comparably secure'' for each number of users, in the sense that the discrete logarithms that have to be computed in an attack are associated to isomorphic groups. We will consider some explicit examples.

 \subsection{Organization of the paper}
 The article is organized in the following way. In Section \ref{sec:notation}, we give basic definitions and facts from finite and profinite group theory that will be used throughout the paper. Section \ref{sec:multilinearBIG} is devoted to the discussion of multilinearity in cryptography and algebra (Section \ref{sec:multilinearity}) as well as the introduction of a key exchange protocol based on multilinear maps in groups (Section \ref{sec:ProtocolI}). In Section \ref{sec:analysis-generic} we discuss the security of the key exchange protocol in the case of generic $p$-groups, while in Section \ref{sec:examples} we discuss some concrete examples.

\begin{thx}
The authors are very thankful to Bernd Sturmfels for putting them in contact and to Pooya Farshim for his useful comments on an early version of this manuscript. 
They are grateful to the University of York, Department of Computer Science, for their hospitality during the time where this work was initiated. 
They also wish to thank the anonymous referees for their helpful feedback. 
\end{thx}

\section{Preliminaries and notation}\label{sec:notation}

The current section includes definitions of and intuition surrounding the mathematical objects playing a role in this paper. Our exposition will be rather synthetic, but we refer the interested reader to
 \cite{Isaacs/08} or \cite{Huppert/67} for more on finite groups and to \cite{Wilson/98} or \cite{RibbZal/10} for more on profinite groups.
 
 Given two groups $G$ and $H$, we will denote by $\Maps(G,H)$ the collection of all functions $G\rightarrow H$ equipped with the group structure inherited from $H$.

\subsection{Commutators and nilpotency}

Let $G$ be a finite group, $X$ a subset of $G$, and $g$ an element in $G$. We will use the standard notations
$|G|$, $|g|$, $\gen{X}$, $\Cyc_G(H)$, to denote the order of $G$, the order of $g$, the subgroup of $G$ generated by $X$, and the centralizer of $H$ in $G$, respectively. 
For elements $x,y\in G$, the \emph{commutator of $x$ with $y$} is 
$$
[x,y]=xyx^{-1}y^{-1} \textup{ implying that } xyx^{-1}=[x,y]y.
$$ 
In particular, conjugation is a left action and, for $x_1,\ldots,x_n\in G$, we group commutators in the following way
\begin{equation}\label{eq:commgrouping}
[x_1,\ldots,x_n]=[x_1,[x_2,[\ldots,[x_{n-1},x_n]]\ldots]].
\end{equation}
If $H$ and $K$ are subgroups of $G$, the \emph{commutator of $H$ and $K$} is the subgroup
\[
[H,K]=\gen{[h,k] \mid h\in H, k\in K}.
\]
We denote, moreover, by $(\gamma_i(G))_{i\geq 1}$ the \emph{lower central series of $G$}, which is recursively defined by
\[
\gamma_1(G)=G \textup{ and } \gamma_{i+1}(G)=[G,\gamma_i(G)],
\]
and by $(\ZG_i(G))_{i\geq 0}$ the \emph{upper central series of $G$}, which is recursively defined by
\[
\ZG_0(G)=\graffe{1} \textup{ and } \ZG_{i+1}(G)/\ZG_i(G)=\ZG(G/\ZG_i(G)).
\]
The \emph{(nilpotency) class of $G$} is $c-1$ where $c$ is the smallest index for which $\gamma_c(G)=1$; equivalently, the class of $G$ is the smallest index $c$ for which $\ZG_c(G)=G$ \cite[Ch.\ 4A]{Isaacs/08}.  In some sense, the class of a group measures ``how far the group is from being abelian''.
If the class of $G$ is finite, then $G$ is said to be \emph{nilpotent};
for a wide overview of nilpotent groups, we refer to \cite[Ch.\ III]{Huppert/67}. An important subclass of that of nilpotent groups is the class of prime power order groups. For a fixed prime $p$, we will call $G$ a \emph{$p$-group} if the order of $G$
 is equal to a power of $p$.

\subsection{Profinite groups}
\label{sec:profinite}

\noindent
Let $(I,\leq )$ be a \emph{directed partially ordered set}, i.e. $\leq$ is a partial order on the set $I$, where the following additional condition is satisfied: 
\[
\textup{ for each } i,j\in I \textup{ there exists } k\in I \textup{ such that } i,j\leq k.
\]
An \emph{inverse} or \emph{projective system of finite groups over $I$} consists of the following data: \begin{itemize}
    \item a collection $\graffe{G_i\mid i\in I}$ of finite groups equipped with the discrete topology, 
    \item for each $i,j\in I$ with $i\geq j$, a homomorphism $\phi_{ij}:G_i\rightarrow G_j$ such that, whenever $i\geq j\geq k$ in $I$, the following diagram commutes
    \[
    \xymatrix{
    G_i \ar[dr]_{\phi_{ij}}\ar[rr]^{\phi_{ik}}&  & G_k \\
        & G_j \ar[ur]_{\phi_{jk}}&
    }
    \]
\end{itemize}
The \emph{inverse} or \emph{projective limit} of the inverse system $(G_i, \phi_{ij})_I$ is the set
\begin{equation}\label{eq:projlim}
    \varprojlim_{i\in I} G_i = \graffe{(g_i)_{i\in I}\in \prod_{i\in I}G_i \mid \textup{ for each } i\geq  j \textup{ one has } \phi_{ij}(g_i)=g_j }
\end{equation}
endowed with the relative product topology. It is not difficult to show that the set in \eqref{eq:projlim} is a group and any such group is called a \emph{profinite group}. An analogous construction for rings, yields the definition of \emph{profinite rings}.

It is clear from their definition that profinite groups are topological groups and it is a standard fact that they are compact, Hausdorff, and totally disconnected \cite[Thm.\ 2.1.3]{RibbZal/10}. If all the groups in the projective system are $p$-groups, then the associated profinite group is called a \emph{pro-$p$ group}. All finite $p$-groups are pro-$p$ groups, while an example of an infinite pro-$p$ group is the underlying additive group in the profinite ring $\Z_p$ of \emph{$p$-adic integers}. Indeed, one has that
\[
\Z_p=\varprojlim_{n\geq 0}\Z/(p^n)
\]
where the projective system is indexed by the nonnegative integers and the maps are, for $m\leq n$, the natural projections $\Z/(p^n)\rightarrow\Z/(p^m)$. In general, if all groups in the projective system satisfy some property $\mathcal{P}$, then the resulting profinite group is called a \emph{pro-$\mathcal{P}$ group}: examples are procyclic groups, proabelian groups, pronilpotent groups, etc.
\vspace{10pt}\\
\noindent
We close Section \ref{sec:profinite} with a synthetic collection of properties of profinite groups and their subgroups. To this end, for a subset $X$ of a profinite group $G$, we denote by $\cl(X)$ the closure of $X$ in the profinite topology of $G$ and we say that $G$ is \emph{topologically generated by $X$} if $G$ equals the closure of the abstract subgroup generated by $X$, i.e.\ $G=\cl(\gen{X})$. It is not difficult to show that, if $X$ is a subgroup of a profinite group $G$, then $\cl(X)$ is itself a profinite subgroup of $G$ \cite[Prop.~2.2.1(a)]{RibbZal/10}. 
For this reason, characteristic subgroups in finite groups are defined, in the profinite setting, by taking the closure of their finite (abstract) analogues. For example, the lower central series of a profinite group $G$ is defined recursively by
\[
\gamma_1(G)=G \textup{ and } \gamma_{i+1}(G)=\cl([G,\gamma_i(G)]).
\]
In particular, pro-$p$ groups are necessarily pronilpotent (equivalently all of their finite quotients are nilpotent) but not necessarily nilpotent: we refer the reader to Section \ref{sec:multi-pro} for an application of this observation.

\section{Multilinearity and NIKE}\label{sec:multilinearBIG}

Since the seminal paper by Boneh and Silverberg \cite{BonehSilverberg/03}, multilinear maps have become very popular in cryptography; see for instance \cite{survey/16} for an overview. In this section, we define what \emph{multilinearity} means in the context of this paper and present a key exchange protocol introduced in \cite{KTT/20}.

\subsection{Multilinearity in algebra and cryptography}\label{sec:multilinearity}

Let $n\geq 2$ be an integer. Classically, a \emph{cryptographic $n$-multilinear map} is a map 
\[
e: G^n \longrightarrow T, \quad (g_1,\ldots,g_n) \longmapsto e(g_1,\ldots, g_n)
\]
for which the following are satisfied:
\begin{enumerate}[label=$(\textup{M}\arabic*)$]
 \item $G$ and $T$ are cyclic groups in which products and inverses are \emph{efficiently computable}, 
 \item $e$ is \emph{efficiently computable} and $\Z$-linear in each component: for each tuple $(a_1,\ldots,a_n)$ of integers and elements $g_1,\ldots,g_n\in G$, one has $$e(g_1^{a_1}, \ldots, g_n^{a_n})=e(g_1,\ldots,g_n)^{a_1\cdots a_n},$$
 \item there is no efficient algorithm to compute discrete logarithms in $G$,
 \item $e$ is \emph{non-degenerate} in the sense that, if $g$ is a generator of $G$, then $e(g,\ldots, g)$ is a generator of $T$;
\end{enumerate}
see for example \cite[\S 2]{BonehSilverberg/03}. 
Producing cryptographic multilinear maps is far from being an easy task given the dichotomy between (M2) and (M3): as of today and to the best of our knowledge, there are indeed no known cryptographic multilinear maps in the sense of Boneh and Silverberg. 

In the present paper, we use the terms `multilinear' and `non-degenerate' in the algebraic sense. As a consequence of our Proposition \ref{prop:KeyPropocolI} the relevant properties of the maps we will be aiming for will be equivalents of (M1)--(M4).
The following conventions will be adopted. Let $G_1,\ldots, G_n, T$ be groups and, for each $i\in\graffe{1,\ldots, n}$, define 
$G^{-i}=G_1\times\ldots\times G_{i-1}\times G_{i+1}\times \ldots\times G_n$. 
A map $e:G_1\times\ldots \times G_n\rightarrow T$ is said to be \emph{$n$-multilinear} (or simply \emph{multilinear}) if, for each integer $i\in\graffe{1,\ldots,n}$ and each tuple ${\bf x}=(x_1,\ldots,x_{i-1},x_{i+1},\ldots,x_n)\in G^{-i}=G_1\times\ldots\times G_{i-1}\times G_{i+1}\times \ldots\times G_n$
, the map 
\[
e_{\bf x}^{(i)}: G_i\longrightarrow T, \quad g \longmapsto e(x_1,\ldots,x_{i-1},g,x_{i+1},\ldots,x_n),
\] 
is a group homomorphism. 
Equivalently, $e$ is multilinear if it is a homomorphism componentwise. We call a map \emph{bilinear} or \emph{trilinear} if it is $2$-linear or $3$-linear, respectively.
Observe now that, for each $i\in\graffe{1,\ldots, n}$, a multilinear map $e:G_1\times \ldots\times G_n\rightarrow T$ has the property that 
\begin{align}\label{eq:def-e_i}
e_i:G_i &\longrightarrow \Maps(G^{-i}, T), \textup{ defined by} \\ \nonumber
g & \longmapsto e_i(g): {\bf x}=(x_1,\ldots,x_{i-1},x_{i+1},\ldots,x_n) \longmapsto e(x_1,\ldots,x_{i-1},g,x_{i+1},\ldots,x_n),
\end{align} 
is a group homomorphism.
 A map $e:G_1\times\ldots\times G_n\rightarrow T$ is called \emph{non-degenerate} if it is multilinear and, for each $i\in\graffe{1,\ldots,n}$, the homomorphism $e_i$ is injective, equivalently, for each $g\in G_i$, one has 
\[
e(G_1\times\ldots\times G_{i-1}\times \graffe{g}\times G_{i+1}\times \ldots\times G_n)=\graffe{1} \Longrightarrow g=1.
\] 
For example, if $n=2$ and $e$ is non-degenerate, then the induced injective homomorphisms $e_1$ and $e_2$ are actually injective homomorphisms 
\[
G_1\longrightarrow \Hom(G_2,T) \textup{ and } G_2\longrightarrow \Hom(G_1,T),
\]
respectively. The following proposition collects a number of classical properties linking nilpotent groups to multilinear maps.

\begin{proposition}\label{prop:multilinear}
Let $G$ be a nilpotent group of class $n\geq 2$. Then the following hold:
\begin{enumerate}[label=$(\arabic*)$]
 \item the commutator map $G\times\gamma_{n-1}(G)\rightarrow \gamma_n(G)$ is bilinear;
 \item the commutator map $G^n\rightarrow \gamma_n(G)$ is multilinear;
 \item if $n=2$, then the map $G/\ZG(G)\times G/\ZG(G)\rightarrow \gamma_2(G)$ that is defined by
 \[
(g_1\ZG(G),g_2\ZG(G)) \longmapsto [g_1,g_2] 
 \]
 is non-degenerate.
\end{enumerate} 
\end{proposition}

\begin{proof}
It is not difficult to show that, if $H$ and $K$ are subgroups of $G$, then the map
\[
H\times K \longrightarrow [H,K], \quad (h,k) \longmapsto [h,k],
\]
is bilinear if and only if $[H,K]$ is central in the group $\gen{H,K}$ generated by $H$ and $K$.

(1) This clearly follows from the fact that $\gamma_{n+1}(G)=1$ equivalently $\gamma_n(G)$ is central in $G$.

(2) We work by induction on $n$ and write $e^n$ for the commutator map $G^n\rightarrow \gamma_n(G)$, where commutators are grouped as in \eqref{eq:commgrouping}. Now, the base of the induction corresponds to $n=2$ and is equivalent to (1). Assume now that $n>2$ and that $e^{n-1}$ is multilinear. Observing that, for each $(x,{\bf y})\in G\times G^{-1}$, one has $e^n(x, {\bf{y}})=[x,e^{n-1}({\bf y})]$ and combining it with (1) yields the claim. 

(3) This is a combination of (1) and the fact that $\ker e^{2}_1=\ker e^{2}_2=\ZG(G)$.
\end{proof}

\subsection{Protocol I: a Non-Interactive Key Exchange}\label{sec:ProtocolI}

Let $n\geq 2$ be an integer. If $n=2$, Mahalanobis and Schinde gave a way for $3$ users to publicly share a secret key in the form of a commutator in a group of class $2$ \cite{MahShinde/17}. In \cite[Protocol~I]{KTT/20}, the authors generalize the last construction to allow $n+1$ users to share a key in the form of a commutator of weight $n$ in a group of class $n$ and, in \cite{KTT/21}, they address the issues of efficiency and security for certain families of groups. We here briefly recall the main steps of Protocol I. To this end, let $G$ be a nilpotent group of class $n$. 
\begin{itemize}
\item Public info: $g_1,\ldots,g_{n}\in G$.
\item Users: $\cor{A}_1,\ldots,\cor{A}_{n+1}$, each of which has chosen an element $a_i\in\Z$.
\item Private keys: $a_1,\ldots,a_{n+1}$.
\item Publicly shared data: $g_i^{a_j}$ where $i\in\graffe{1,\ldots, n}$ and $j\in\graffe{1,\ldots,n+1}$.
\item Shared secret key: $[g_1,\ldots,g_n]^{a_1\cdots a_{n+1}}$, which each user can compute from the shared data and its private key using multilinearity of the commutator map.
\end{itemize}
For example, $\cor{A}_1$ can recover the shared key via computing
\[
[g_1^{a_{n+1}},g_2^{a_2}, \ldots,g_n^{a_n}]^{a_{1}}=[g_1,\ldots,g_n]^{a_1a_2\cdots a_{n+1}}.
\]
Call now $c=[g_1,\ldots, g_n]$ and $a=a_1\cdots a_{n+1}$. The security of Protocol~I  is connected to the difficulty of recovering, from the public information, any of the $a_j$'s or $a$. Indeed, for any $j\in \graffe{a_1,\ldots, a_{n+1}}$, an eavesdropper can use the public information to compute
\begin{enumerate}
\item $c=[g_1,\ldots,g_n]$ from $g_1,\ldots,g_n$,
\item $c^{a/a_j}$ for example in the following way:
\[
c^{a/a_j}=\begin{cases}
[g_1^{a_{n+1}},g_2^{a_2},\ldots ,a_n^{a_n}]& \textup{if } j=1, \\
[g_1^{a_1},\ldots, g_{j-1}^{a_{j-1}}, g_j^{a_{n+1}},g_{j+1}^{a_{j+1}},\ldots, a_n^{a_n}] & \textup{if } 1<j<n+1, \\
[g_1^{a_1},g_2^{a_2},\ldots, g_n^{a_n}] & \textup{if } j=n+1.
\end{cases}
\]
\end{enumerate}
We will, in Section \ref{sec:analysis-generic}, concretize what ``recovering $a_j$'' actually means in the context of Protocol I and show how, in some sense, the choice of $j$ is irrelevant (see Proposition \ref{prop:KeyPropocolI}). The next example shall serve as a first hint in this direction.

\begin{example}\label{ex:class2}
Assume $n=2$, equivalently $G$ is nilpotent of class $2$ and $\gamma_2(G)$ is central. 
Then, by Proposition \ref{prop:multilinear}, the commutator map $G\times G\rightarrow \gamma_2(G)$ is bilinear and induces the non-degenerate map 
\[
e:G/\ZG(G)\times G/\ZG(G)\rightarrow\gamma_2(G), \quad (x\ZG(G),y\ZG(G))\mapsto [x,y].
\]
In particular, if $\alpha\in\Z$ and $\overline{g_i}=g_i\ZG(G)$, then we have
\[
c^{\alpha \bmod |c|} = c^{\alpha} = [g_1,g_2]^{\alpha}=e(\overline{g_1},\overline{g_2})^{\alpha}=e(\overline{g_1}^{\alpha\bmod |\overline{g_1}|},\overline{g_2})=e(\overline{g_1},\overline{g_2}^{\alpha\bmod |\overline{g_2}|}).
\]
A consequence of the last series of equalities is that, if $j\in\graffe{1,2,3}$ and an Eavesdropper can establish $a_j\bmod\min\graffe{|c|,|\overline{g_1}|,|\overline{g_2}|}$, then they can also determine $c^a$.
\end{example}

\noindent
We conclude the present section extending Example~\ref{ex:class2} to the case of trilinear maps, cf.\ Example~\ref{ex:class3}. The following lemma can be found in any standard text in group theory; see e.g.\ \cite[Cor.\ 4.10]{Isaacs/08}.

\begin{lemma}[Three subgroups lemma]\label{lemma:3sbgs}
Let $G$ be a group and let $N$ be a normal subgroup of $G$. Let moreover $A,B,C$ be subgroups of $G$ such that $[A,B,C]$ and $[B,C,A]$ are contained in $N$. Then $[C,A,B]$ is contained in $N$.
\end{lemma}

\begin{example}\label{ex:class3}
Assume $n=3$, equivalently $G$ is nilpotent of class $3$ and $\gamma_3(G)$ is central. 
Then the commutator map $G\times G \times G\rightarrow \gamma_3(G)$ is trilinear and we claim that it induces a non-degenerate map 
\[
e:G/\Cyc_G(\gamma_2(G))\times G/\ZG_2(G) \times G/\ZG_2(G)\rightarrow\gamma_3(G).
\]
For our claim to hold true, we need to show that 
\[
[\Cyc_G(\gamma_2(G)),[G,G]]=[G,\ZG_2(G),G]=[G,G,\ZG_2(G)]=1
\]
and that $\Cyc_G(\gamma_2(G))$ and $\ZG_2(G)$ are maximal with the above properties. The claim is easily settled for $\Cyc_G(\gamma_2(G))$ as, by definition,  $\Cyc_G(\gamma_2(G))$ is the largest subgroup of $G$ centralizing $\gamma_2(G)=[G,G]$. We look now at $\ZG_2(G)$ and we call $K$ the second kernel of $e$, i.e.\ the largest subgroup of $G$ such that $[G,K,G]=1$. Note that $K=\ker e_3$, where $e_3$ is as in \eqref{eq:def-e_i}. Now, $[G,[K,G]]=1$ is equivalent to $[K,G]$ being central and so $K$ is maximal with the property that $[K,G]\subseteq \ZG(G)$. It follows from the definition of the upper central series that $K=\ZG_2(G)$. Since $[G,\ZG_2(G)]=[\ZG_2(G),G]$, a symmetric argument yields the claim. 

We next note that $\ZG_2(G)$ is contained in $\Cyc_G(\gamma_2(G))$. Indeed, since $[G,[G,\ZG_2(G)]]=[G,[\ZG_2(G),G]]=1,$ Lemma~\ref{lemma:3sbgs} yields that $[\ZG_2(G),[G,G]]=1$. 
It follows, in particular and in analogy to Example~\ref{ex:class2}, that, if $\overline{G}=G/\Cyc_G(\gamma_2(G))$, the determination of the shared key from the public information can be achieved via the determination of 
$a_j\bmod \min\{|c|,|\overline{g_1}|, |\overline{g_2}|, |\overline{g_3}|\}$.
\end{example}

We remark that, for any $n$, the multilinear maps that end up playing the real role in this context are multilinear maps on abelian groups, i.e.\ $\Z$-modules; in Example \ref{ex:class2} the abelian group is $G/\ZG(G)$ and in Example \ref{ex:class3} it is $G/\Cyc_G(\gamma_2(G))$.

\subsection{Degenerations and NIKE}\label{sec:deg&nike}
The key exchange protocol presented in Section \ref{sec:ProtocolI} is non-interactive, i.e.\ it enables the users, each of which has a private key, to agree on a symmetric shared key without any interaction. Probably the most known Non-Interactive Key Exchange scheme, in short NIKE, is the one by Diffie and Hellman \cite{DH/76} over cyclic groups. For more on NIKE, we refer the reader to \cite{NIKE}.

Within the context of non-interactive key exchange, we remark that, if $G$ is a finite nilpotent group of class $n\geq 2$, then degenerations of Protocol I yield symmetric key exchange schemes for any number of users between $2$ and $n+1$. Indeed, in the notation of Section \ref{sec:ProtocolI}, if $\cor{B}$ is a subset of $\graffe{\cor{A}_1,\ldots,\cor{A}_{n+1}}$ of cardinality at least $2$, then a shared key for the users in $\cor{B}$ can be computed by each user via assuming $$\graffe{a_1,\ldots,a_{n+1}}\setminus \graffe{a_i \mid \cor{A}_i\in\cor{B}}=\graffe{1}.$$ For example, if $n=5$ and $\cor{B}=\graffe{\cor{A}_1,\cor{A}_2}$, then a shared key for the users in $\cor{B}$ can be computed in the following way: 
\begin{itemize}
    \item $\cor{A}_1$ computes $[g_1^{a_2},g_2,\ldots,g_5]^{a_1}$ while
    \item $\cor{A}_2$ computes $[g_1^{a_1},g_2,\ldots,g_5]^{a_2}$.
\end{itemize}
We remark that, the security analysis in these degenerate cases is analogous to the one made in Section \ref{sec:analysis-generic} for Protocol I.

\section{Security in the generic case}\label{sec:analysis-generic}

The purpose of this section is to discuss the security of Protocol I from Section \ref{sec:ProtocolI} for generic groups. 
We will show that, if $G$ is a finite $p$-group and the element $c=[g_1,\ldots,g_n]$ has order $p^{\alpha}$,  then recovering the shared key $c^a$ from Protocol I from the public data is as hard as solving the CDH in a cyclic group of order $p^\alpha$, which is, in some cases, known to be as hard as solving the DLP; cf.\ \cite{denBoer/90,Maurer/94}.
In the second part of this section, we will employ pro-$p$ groups to produce multilinear maps of any rank. We will discuss security in this case, too.

\subsection{Reduction to cyclic groups}\label{sec:reduction}

In the present section, we discuss the security of Protocol I in the case of \emph{generic} $p$-groups, in the sense of \cite{Sutherland/11}; a more detailed reference is \cite[Ch.~1]{Sutherland/PHD}. The main result of this section is Proposition \ref{prop:KeyPropocolI} and the following example shall serve as a warm-up towards proving it.


\begin{example}\label{ex:Cclass2}
Let $p$ be a prime number, $\alpha$ a positive integer, and $C$ a cyclic group of order $p^\alpha$ equipped with a non-degenerate alternating map $e:C\times C\rightarrow  C$. Let $G$ be the group $C\times C\times C$ where the operation is given by 
\[
(x,y,z)(x',y',z')=(xx',yy',zz'e(x,y')).
\]
Then $G$ is a group of class $2$ satisfying $G/\ZG(G)\cong C\times C$ and $\gamma_2(G)\cong C$. Let now $g_1,g_2\in G$ be such that $c=[g_1,g_2]$ generates $\gamma_2(G)$ 
and note that $|g_1|=|g_2|=|[g_1,g_2]|=p^\alpha$. Let $\cor{A}_1,\cor{A}_2,\cor{A}_3$ be three users, each of which has chosen a secret element $a_i\in(\Z/(p^\alpha))^*$. The following information is shared publicly:
\begin{itemize}
    \item $\cor{A}_1$ shares $g_1^{a_1}$,
    \item $\cor{A}_2$ shares $g_2^{a_2}$,
    \item $\cor{A}_3$ shares $g_1^{a_3}$ and $g_2^{a_3}$.
\end{itemize}
Then each user can compute the shared key $[g_1,g_2]^{a_1a_2a_3}$. An eavesdropper can then compute
\[
c=[g_1,g_2],\ c^{a_1},\ c^{a_2},\  c^{a_3},\  c^{a_1a_2},\  c^{a_2a_3},\  c^{a_1a_3}
\]
and compute the shared key $c^a=c^{a_1a_2a_3}$ via, for instance, solving the discrete logarithm for $c^{a_1}$ and $c^{a_2a_3}$ with respect to $c$. 
\end{example}

\noindent
The following result shows that Example \ref{ex:Cclass2} is a particular instance of a more general phenomenon. We invite the reader to compare Proposition \ref{prop:KeyPropocolI} also with Examples \ref{ex:class2} and \ref{ex:class3}.

\begin{proposition}\label{prop:KeyPropocolI}
Let $p$ be a prime number, $G$ a finite generic $p$-group of class $n$, and $g_1,\ldots, g_n$ elements of $G$. Define $c=[g_1,\ldots,g_n]$ and let $p^\alpha$ denote the order of $c$. Then there exist a cyclic group $\tilde{G}$ of order $p^{\alpha}$, a non-degenerate multilinear map $\tilde{e}:\tilde{G}^n\rightarrow C$, and surjective maps $\phi_i:\gen{g_i}\rightarrow\tilde{G}$ such that $c=[g_1,\ldots,g_n]=\tilde{e}(\phi_1(g_1),\ldots,\phi_n(g_n))$.
\end{proposition}

\begin{proof}
For each $i\in\graffe{1,\ldots,n}$ denote $G_i=\gen{g_i}$ and, for $e:G_1\times\ldots\times G_n\rightarrow\gamma_n(G)$ denoting the restriction of the commutator map from Proposition \ref{prop:multilinear}(2), let $K_i$ be the kernel of $e_i$ as defined in \eqref{eq:def-e_i}. In particular, $K_i\subseteq G_i$ is a subgroup, maximal with the property that $[G_1,\ldots,G_{i-1},K_i,G_{i+1},\ldots, G_n]=1$. Set, moreover, $C=\gen{c}$ and write $|C|=p^{\alpha}$. Then $e$ induces a non-degenerate multilinear map
\[
\tilde{e}:G_1/K_1\times G_2/K_2 \times \ldots \times G_n/K_n\longrightarrow C.
\]
We now claim that, for each $i$, one has $|G_i:K_i|=|C|=p^{\alpha}$. For this, fix $i\in\graffe{1,\ldots, n}$ and $m\in\Z$. From the multilinearity of $\tilde{e}$ it follows that 
\[
\tilde{e}(g_1K_1,\ldots,g_{i-1}K_{i-1},g_i^{m}K_i,g_{i+1}K_{i+1},\ldots g_nK_n)=[g_1,\ldots, g_n]^m=c^{m}. 
\]
Choosing $m=p^{\alpha}$, we derive that $g_i^m\in K_i$ and so that $|G_i:K_i|\leq p^{\alpha}$ while, choosing $m=|G_i:K_i|$, we conclude that $c^{|G_i:K_i|}=1$ and thus $p^{\alpha}\leq |G_i:K_i|$. 

We have proven that all $G_i/K_i$ are mutually isomorphic and isomorphic to a cyclic group $\tilde{G}$ of order $p^{\alpha}$. Identifying all quotients $G_i/K_i$ with $\tilde{G}$ and letting, for each $i$, the map $\phi:G_i\rightarrow G_i/K_i$ being the canonical projection yields the claim.
\end{proof}

\noindent
We remark that, thanks to the discussion in the proof of Proposition \ref{prop:KeyPropocolI} and the fact that $p$-groups are polycyclic, one sees that Protocol I and its degenerations are graded encoding schemes \cite[\S~1.3]{survey/16}.

\subsection{The connection to CDH}\label{sec:CDH}

The Computational Diffie-Hellman problem, in short CDH, asks to determine, given a cyclic group $C$ and a triple $(c,c^x,c^y)$ with $c\in C$ and $x,y\in\Z$, the element $c^{xy}$. The security of the Diffie-Hellman key exchange \cite{DH/76} links to the hardness of CDH, which itself reduces to DLP. As of today and to the best of our knowledge, however, it is not know whether CDH and DLP are equivalent, though it is proven to be the case sometimes, cf.\ \cite{denBoer/90,Maurer/94}. The next proposition relates the security of Protocol~I to CDH and DLP.

\begin{proposition}\label{prop:CDH}
Let $p$ be a prime number, $G$ a finite generic $p$-group of nilpotency class $n$, and $g_1,\ldots, g_n$ elements of $G$. Define $c=[g_1,\ldots,g_n]$. 
If \emph{CDH} is hard in the cyclic $p$-group $\gen{c}$ generated by $c$, then Protocol~I is secure in the presence of an eavesdropper. 
\end{proposition}

\begin{proof}
Let $\cor{A}$ be an adversary: we show that, if $\cor{A}$ can compute the shared key from the public information, then $\cor{A}$ can solve the CDH in the cyclic $p$-group $\gen{c}$ generated by $c$.  
Observe that, if $a_1,\ldots,a_{n+1}$ are the elements in $\Z$ chosen by the users $\cor{A}_1,\ldots,\cor{A}_n$, then, from the following (publicly shared) information
\begin{itemize}
    \item $g_1,\ldots, g_n$,
    \item $g_i^{a_j}$ for $i\in\graffe{1,\ldots,n}$, $j\in\graffe{1,\ldots, n+1}$,
\end{itemize}
 and using multilinearity, $\cor{A}$ can compute $c=[g_1,\ldots,g_n]$ and $c^{a_1\cdots a_n}$ and $c^{a_{n+1}}$. 
Moreover, by assumption, the adversary $\cor{A}$ is also able to retrieve $c^a$ from the public information. 

In particular, setting $a_1=x$, $a_2=\ldots=a_n=1$, and $a_{n+1}=y$, the adversary $\cor{A}$ is able to compute 
\[
c=[g_1,\ldots,g_n], \quad c^x=c^{a_1\cdots a_n}, \quad c^y=c^{a_{n+1}}
\]
and, consequently, also $c^a=c^{xy}$  from the publicly shared information from Protocol~I. 
\end{proof}

\begin{example}
Let $p$ be a prime number and let $G$ be a group of maximal class, equivalently, if the class of $G$ is $n$, then the order of $G$ is $p^{n+1}$. It is then an easy exercise to show that $|\gamma_n(G)|=p$ and thus, thanks to Proposition~\ref{prop:CDH},  an eavesdropper who can solve the DLP in a cyclic group of order $p$, can also compute the shared key from Protocol~I from the publicly shared information.
\end{example}

\noindent
We note that, since for each pair $(p,n)$, there exists a $p$-group of maximal class $n$ (for example as a consequence of \cite[Thm.~4.3]{Blackburn/58}), the class of $p$-groups of maximal class provides platforms for Protocol~I for an arbitrary number of users in which the cost of recovering the key from the public information is likely to be comparable to solving the DLP in a cyclic group of order $p$.

\subsection{Algorithms for the discrete logarithm problem}\label{sec:algorithms}

In this section, we give a short overview of the existing algorithms solving the discrete logarithm problem in generic (cyclic and) abelian $p$-groups and of their computational costs. To this end, let $p$ be a prime number, $\alpha$ a positive integer and $C$ a generic cyclic group of order $p^\alpha$. 

With the Pohlig-Hellman algorithm \cite{PohHell/78}, the discrete logarithm problem in $C$ can be solved in $O(\alpha\sqrt{p})$ group operations. Teske's algorithm \cite{Teske/99} for solving the DLP in abelian groups matches the Pohlig-Hellman algortihm in complexity in the case of cyclic groups. 
 An improvement on Teske's algorithm is given in \cite{Sutherland/11} and, though the computational costs for the case of cyclic groups seem not to be lower than the ones computed in \cite{PohHell/78}, discrete logarithms are computed in a faster way in practice with Sutherland's algorithm (see for example \cite[\S 6]{Sutherland/11}). 
 
 In general and to the best of our knowledge, it appears therefore that, if $C\cong\gen{[g_1,\ldots,g_n]}$ for $g_1,\ldots, g_n$ are as in Protocol I, the theoretical cost of breaking Protocol I is expected to be $O(\alpha\sqrt{p})$. Indeed, in generic groups, all multiplication costs in $C$ or $\Z/(p^{\alpha})^*$ are negligible compared to solving the DLP in $C$.

\subsection{Multilinear maps from profinite groups}\label{sec:multi-pro}
We have seen, in Section \ref{sec:reduction}, that there is a family of $p$-groups, namely the groups of maximal class, providing multilinear maps of any rank and thus allowing a Diffie-Hellman like key exchange between any number of users. We have, moreover, seen that, in a generic group of maximal class, the theoretical complexity of breaking Protocol I is expected to be $O(\sqrt{p})$.
We now observe that, if $G$ is a finite $p$-group of class $n$ and $1\leq m\leq n$ is an integer, then $\overline{G}=G/\gamma_m(G)$ is a finite $p$-group of class $m-1$ and thus $\overline{G}$ can be used for a key exchange protocol between $m$ users. In this sense, a finite $p$-group of class $n$ can be used in Protocol I for any number of users not exceeding $n+1$. It would be convenient if we had, at our disposal, an infinite pro-$p$ group with finite quotients of any nilpotency class. Even better it would be if we could realize such groups in such a way that quotients of consecutive elements of the lower central series can achieve arbitrary exponent (yielding increasing levels of security). Fortunately, such groups exist and, in the next and last section of this paper, we will look at some concrete examples. We remark that, in such examples, given the intrinsic linearity of the groups in question, Protocol I turns out to be much less secure than the expected theoretical complexity discussed in Section \ref{sec:algorithms}.




\section{Non-generic examples}\label{sec:examples}

In this section we discuss two concrete families of groups giving rise to key exchange protocols for $3$ and any number of users, respectively. In the context of (M1)--(M4) from Section \ref{sec:multilinearity}, we will see how, within these families, the linear nature of the groups contributes to the efficiency (M2), but penalizes security (M3).

\subsection{Extraspecial groups}

As mentioned in Section \ref{sec:ProtocolI}, Protocol I is a generalization of the key exchange presented in \cite{MahShinde/17} to any nilpotency class and, thus, to any number of users. In the last paper, the authors discuss possibilities for ``good platform groups'' for their cryptosystem. The outcome of the investigation does not yield cryptographic bilinear maps in the sense of \cite{BonehSilverberg/03} mainly because of the incompatibily of (M2) and (M3) in the considered examples. In \cite{MahShinde/17}, a special emphasis is put on extraspecial groups of exponent $p^2$; in the next example we look at extraspecial groups of exponent $p$ (the only other possibility for the exponent of an extraspecial group). For more detail on extraspecial groups, we refer to \cite[Ch.\ III.13]{Huppert/67}.   

\begin{example}
Let $p$ be an odd prime number and let $m$ be a positive integer. For elements ${\bf u,v}\in\F_p^m$ we write ${\bf u}\cdot{\bf v}$ for their (scalar) product, i.e.\ ${\bf u}\cdot{\bf v}={\bf u}{\bf v}^{\mathrm{T}}$. Observe that the last product defines a bilinear map $\F_p^m\times\F_p^m\rightarrow \F_p$. Let $G=\Heis_{2m+1}(\F_p)$ be the group with ground set $\F_p^m\times \F_p^m\times \F_p$ and multiplication defined by 
\[
({\bf u},{\bf v},z)({\bf u}', {\bf v}',z')=({\bf u}+{\bf u}',{\bf v}+{\bf v}',z+z'+{\bf u}\cdot{\bf v}').
\] 
The group $G$ has order $p^{2m+1}$, class $2$ with $\gamma_2(G)\cong \F_p$, and exponent $p$.
Moreover, $G$ can also be seen as a group of upper unitriangular matrices in the following sense:  
\[
G=\graffe{ 
M({\bf u},{\bf u},z)=\begin{pmatrix}
1 & {\bf u} & z \\
{\bf 0}^{\mathrm{T}} & \mathrm{Id}_m & {\bf v}^{\mathrm{T}} \\
0 & {\bf 0} & 1
\end{pmatrix} \mid {\bf u},{\bf u}\in\F_p^m,z\in\F_p
}
\]
and the multiplication is the naturally expected one.
If $a$ is an integer, one can easily compute 
\[
M({\bf u},{\bf v},z)^a
=M\Big(a {\bf u},a {\bf v},\binom{a+1}{2}z\Big)\]
and so the shared key from Protocol I can here be computed by any eavesdropper with a cost of at most $O(\log( p)^2)$ group operations, which is quite less secure than $O(\sqrt{p})$. This shows that, in this family of groups, similarly to the case of their siblings of exponent $p^2$, gaining in computational efficiency with a linear representation results in a loss of security.  
\end{example}

\subsection{Profinite examples}
Until the end of this section, let $p>3$ be a prime number and let $\alpha$ be a positive integer. Let, moreover, $\Z_p$ denote the ring of $p$-adic integers and let $t\in\Z_p$ be a quadratic nonresidue modulo $p$. Define $\Delta_p$ to be the quaternion algebra 
$$\Delta_p=\Z_p\oplus\Z_p\mathrm{i}\oplus\Z_p\mathrm{j}\oplus\Z_p\mathrm{k}
\textup{ where }
\mathrm{i}^2=t, \, \mathrm{j}^2=p,\ \text{and} \ \mathrm{k}=\mathrm{ij}=-\mathrm{ji}.$$
The quaternion algebra $\Delta_p$ is equipped with a bar map, defined by 
$$x=a+b\mathrm{i}+c\mathrm{j}+d\mathrm{k} \longmapsto 
\overline{x}=a-b\mathrm{i}-c\mathrm{j}-d\mathrm{k},$$ 
which is an anti-homomorphism of order $2$ and has $\mathfrak{m}=\Delta_p\,\mathrm{j}$ as its unique maximal ideal. It follows that an element $x=a+b\mathrm{i}+c\mathrm{j}+d\mathrm{k}$ belongs to $\mathfrak{m}$ if and only if both $a$ and $b$ belong to $p\Z_p$. Moreover, for each $k\in\Z_{\geq 1}$, the ideal $\mathfrak{m}^k$ is principal generated by $\mathrm{j}^k$ and therefore, for each $s\in\Z_{\geq 0}$, one has 
\[\mathfrak{m}^{2s}=p^s\Delta_p \ \ \text{and} \ \ \mathfrak{m}^{2s+1}=p^s\mathfrak{m}. \]
As a result, for each $k\in\Z_{\geq 1}$, the quotient $\mathfrak{m}^k/\mathfrak{m}^{k+1}$ is a vector space over $\F_p$ of dimension $2$. Now, for each $k\in\Z_{\geq 1}$, the set $1+\mathfrak{m}^k$ is easily seen to be a subgroup of $\Delta_p^*$ and the natural map 
$$(1+\mathfrak{m}^k)/(1+\mathfrak{m}^{k+1})\rightarrow \mathfrak{m}^k/\mathfrak{m}^{k+1}$$  is an isomorphism of groups. It follows that $1+\mathfrak{m}$ is a pro-$p$ subgroup of $\Delta_p^*$.
Define $$\Sn(\Delta_p)=(1+\mathfrak{m})\cap\graffe{x\in\Delta_p : \overline{x}=x^{-1}}.$$ 
Then $\Sn(\Delta_p)$ is a closed subgroup of $1+\mathfrak{m}$ and thus a pro-$p$ group itself. We have here lightened the notation from \cite{GSK/09}, where the group $\Sn(\Delta_p)$ is denoted by $\SL_1^1(\Delta_p)$. With the notation of this paper, the following structural results can be found for example in \cite{Sta/17}.

\begin{lemma}\label{lemma:LCS}
Let $k$ and $\ell$ be positive integers. Then the following hold:
\begin{enumerate}[label=$(\arabic*)$]
    \item one has $\gamma_k(\Sn(\Delta_p))=(1+\mathfrak{m}^k)\cap S(\Delta_p)$ and $|\gamma_k(\Sn(\Delta_p)):\gamma_{k+1}(\Sn(\Delta_p))|\in\graffe{p,p^2}$.
    \item the map ${x\mapsto x^p}$ on $\Sn(\Delta_p)$ induces an isomorphism 
$$\rho_k:\gamma_k(\Sn(\Delta_p))/\gamma_{k+1}(\Sn(\Delta_p))\rightarrow \gamma_{k+2}(\Sn(\Delta_p))/\gamma_{k+3}(\Sn(\Delta_p)).$$
 \item if $k$ is odd, then $\cl([\gamma_k(\Sn(\Delta_p),\gamma_{\ell}(\Sn(\Delta_p)))])=\gamma_{k+\ell}(\Sn(\Delta_p))$.
\end{enumerate}
\end{lemma}

\begin{proof}
(1) The first claim is \cite[Lem.\ 438]{Sta/17} while the second follows from the fact that $|\mathfrak{m}^k:\mathfrak{m}^{k+1}|=p^2$. (2) This is \cite[Lem.\ 439]{Sta/17}. (3) This follows from combining Lemmas 441 and 328(2) from \cite{Sta/17}.
\end{proof}

\begin{proposition}\label{prop:exp}
Denote $G=\Sn(\Delta_p)$.
Let $i=2\alpha -1$ and set $H=\gamma_i(G)$. Then, for each positive integer $k$, the subgroup $\gamma_k(H)$ is open of finite index in $G$ and the exponent of $\gamma_{k}(H)/\gamma_{k+1}(H)$ is $p^{\alpha}.$
\end{proposition}

\begin{proof}
Let $k$ be a positive integer.
As a consequence of Lemma \ref{lemma:LCS}(3), one has
\[
\gamma_k(H)=\gamma_{ki}(G) \textup{ and }
\gamma_{k+1}(H)=\gamma_{(k+1)i}(G)=\gamma_{ki+2\alpha-1}(G).
\]
Now, it follows from Lemma \ref{lemma:LCS}(2) that the exponent of $\gamma_{k}(H)/\gamma_{k+1}(H)$ is equal to $p^\epsilon$ where
\[
\epsilon = \frac{(ki+2\alpha -1 )+1-ki}{2}=\alpha.
\]
The fact that $\gamma_k(H)$ is open of finite index in $G$ follows from Lemma \ref{lemma:LCS}(1). 
\end{proof}

\begin{proposition}\label{prop:powers}
Denote $G=\Sn(\Delta_p)$ and let $k,m$ be positive integers.
Let $i=2\alpha -1$ and set $H=\gamma_i(G)$.
Let $x=a+b\mathrm{i}+c\mathrm{j}+d\mathrm{k}$ be an element of $\gamma_k(H)$. 
Then one has 
\begin{equation}\label{eq:power}
x^m\equiv a^m+m b\mathrm{i}+m c\mathrm{j}+m d\mathrm{k} \bmod \mathfrak{m}^{(k+1)i}. 
\end{equation}
\end{proposition}

\begin{proof}
Thanks to the combination of Lemma \ref{lemma:3sbgs}(1) and (3), we have $\gamma_k(H)=\gamma_{ki}(G)=(1+\mathfrak{m}^{ki})\cap G$ and so we have control on the $p$-adic valuations of the coefficients of $x$. More in detail, three cases can occur: 
\begin{enumerate}[label=$(\arabic*)$]
    \item $ki$ is odd, in which case $(k+1)i$ is even. In this case we have
    \begin{align*}
        (a,b,c,d)\equiv (1,0,0,0) & \bmod (p^{(ki+1)/2}, p^{(ki+1)/2},p^{(ki-1)/2},p^{(ki-1)/2}) \\
        \textup{work} & \bmod (p^{(ki+i)/2}, p^{(ki+i)/2},p^{(ki+i)/2},p^{(ki+i)/2}),
    \end{align*}
        \item $ki$ and $(k+1)i$ are both even. In this case we have
    \begin{align*}
        (a,b,c,d)\equiv (1,0,0,0) & \bmod (p^{ki/2}, p^{ki/2},p^{ki/2},p^{ki/2}) \\
        \textup{work} & \bmod (p^{(ki+i)/2}, p^{(ki+i)/2},p^{(ki+i)/2},p^{(ki+i)/2}),
    \end{align*}
    \item $ki$ is even and $(k+1)i$ is odd, implying that $k$ is even and $i$ odd. In this case we have
    \begin{align*}
        (a,b,c,d)\equiv (1,0,0,0) & \bmod (p^{ki/2}, p^{ki/2},p^{ki/2},p^{ki/2}) \\
        \textup{work} & \bmod (p^{(ki+i+1)/2}, p^{(ki+i+1)/2},p^{(ki+i-1)/2},p^{(ki+i-1)/2}).
    \end{align*}
\end{enumerate}
The equivalences we will write in this proof all follow from the analysis of the different possibilities (1)-(2)-(3), so we will not explicitly write all computations. The equivalences we will use are: 
\[
b^2\equiv 0,\ pc^2\equiv 0,\ pd^2\equiv 0,\ ab\equiv b,\ ac\equiv c,\ ad \equiv d \bmod \mathfrak{m}^{(k+1)i}.
\]
We work by induction on $m$. If $m=1$, the statement is clearly true, so we assume that $m>1$ and that \eqref{eq:power} holds for $m-1$. Defining 
\[y=a^{m-1}+(m-1)b\mathrm{i}+(m-1)c\mathrm{j}+(m-1)d\mathrm{k},\]
we have that $x^{m-1}\equiv y\bmod \mathfrak{m}^{(k+1)i}$. 
Write $xy=A+B\mathrm{i}+C\mathrm{j}+D\mathrm{k}$. Modulo $\mathfrak{m}^{(k+1)i}$, we then compute
\begin{align*}
    A & = a^{m}+t(m-1)b^2+ p(m-1)c^2 -pt(m-1)d^2 \equiv a^m,  \\
    B & = (m-1)ab+a^{m-1}b \equiv mb, \\
    C & = (m-1)ac+a^{m-1}c \equiv mc, \\
    D & = (m-1)ad+a^{m-1}d \equiv md.
\end{align*}
We conclude by observing that
$x^m = xx^{m-1} \equiv xy \equiv A+B\mathrm{i}+C\mathrm{j}+D\mathrm{k}\bmod \mathfrak{m}^{(k+1)i}$.
\end{proof}

\noindent
We close the article with a discussion of the implications of Proposition \ref{prop:powers} on Protocol I. For this, set $i=2\alpha-1$ and $H=\gamma_i(\Sn(\Delta_p))$. Choose a number of users $n+1$. Then Proposition \ref{prop:exp} ensures that $\overline{H}=H/\gamma_{n+1}(H)$ is a finite $p$-group of class $n$ satisfying $\exp(\gamma_n(\overline{H}))=p^{\alpha}$. In other words $\overline{H}$ is a platform group for Protocol I on $n+1$ users. Note now that, thanks to Proposition \ref{prop:exp} and Lemma \ref{lemma:LCS}(1), we have 
\[
\gamma_n(H)=(1+\mathfrak{m}^{ni})\cap\Sn(\Delta_p) \textup{ and }
\gamma_{n+1}(H)=(1+\mathfrak{m}^{(n+1)i})\cap\Sn(\Delta_p),
\]
so the powers of an element $x$ in $\gamma_n(\overline{H})$ are described exactly by the formula from Proposition \ref{prop:powers}. In particular, the cost an eavesdropper has to pay to recover the shared key is at most  $O(\alpha^2(\log p)^2)$ group operations, against the generic $O(\alpha\sqrt{p})$ group operations predicted by the Pohlig-Hellman algorithm.\\

We remark that the examples of profinite groups we considered in this section are a particular instance of so-called $p$-adic analytic groups. Such groups come with the structure of an analytic manifold which seems to favour (M2) but penalize (M3). Our investigation suggests that one should, for security reasons, consider non-analytic pro-$p$ groups in the context of Protocol~I. 




\bibliographystyle{abbrv}
\bibliography{references} 
 
\vspace*{2em}\noindent{\footnotesize\begin{minipage}[t]{0.5\textwidth} Delaram Kahrobaei \\
University of York \\
Department of Computer Science\\ 
Deramore Lane\\ 
YO10 5GH York\\
United Kingdom\\ \quad\\ E-mail: \href{mailto:delaram.kahrobaei@york.ac.uk}{delaram.kahrobaei@york.ac.uk}\end{minipage}
\hfill
\begin{minipage}[t]{0.5\textwidth} Mima Stanojkovski\\ Max-Planck-Institut f\"ur Mathematik in \\ den Naturwissenschaften\\ Inselstrasse 22\\ 04103 Leipzig\\ Germany\\ \quad\\ E-mail: \href{mailto:mima.stanojkovski@mis.mpg.de}{mima.stanojkovski@mis.mpg.de}\end{minipage}}

 \end{document}